\documentclass[imslayout,noinfoline,a4paper,12pt,twoside]{imsart}
\RequirePackage[OT1]{fontenc} 
\RequirePackage{amsthm,amsmath,amssymb}
\usepackage{multirow}

\startlocaldefs
\theoremstyle{plain}

\newtheorem{theo}{Theorem}[section]
\newtheorem{bsp}{Example}[section]

\newtheorem{lem}{Lemma}[section]

\newtheorem{defi}{Definition}[section]

\endlocaldefs

\renewcommand{\[}{\begin{eqnarray*}}
\renewcommand{\]}{\end{eqnarray*}}
\newcommand{\la}{\begin{eqnarray}}
\newcommand{\al}{\end{eqnarray}}

\renewcommand{\epsilon}{\varepsilon}
\renewcommand{\phi}{\varphi}

\newcommand{\N}{{\mathbb N}}

\newcommand{\Z}{{\mathbb Z}} 

\renewcommand{\P}{{\mathbb P}}  

\renewcommand{\theta}{\vartheta}

\begin{document}

\begin{frontmatter}


\title{Rigorous Computing of Rectangle Scan Probabilities for Markov Increments}

\runtitle{Rectangle Scan Probabilities for Markov Increments
}

\begin{aug}
\author{\fnms{Jannis} \snm{Dimitriadis}
\ead[label=e1]{dimitria@uni-trier.de}}

\runauthor{J.~Dimitriadis, \today}

\affiliation{Universit\"at Trier} 
 {\rm \today}\\
 {\footnotesize\tt \jobname.tex}

\address{Universit\"at Trier\\
Fachbereich IV - Mathematik\\ 
54286 Trier \\
Germany\\
\printead{e1}}
\end{aug}

\begin{abstract}
Extending recent work of Corrado, we derive an algorithm that computes rigorous upper and lower bounds for rectangle scan probabilities for Markov increments. We experimentally examine the closeness 
of the bounds computed by the algorithm and we examine the range of tractable input variables. 
\end{abstract}

\begin{keyword}[class=AMS]
\kwd[Primary ]{62P10}.
\end{keyword}
\begin{keyword}
\kwd{Scan Statistics}
\kwd{Rectangle Probabilities}
\kwd{Markov Chain}
\kwd{Interval Arithmetic}.
\end{keyword}

\end{frontmatter}



\section{Introduction}
Let $n$ balls randomly fall into $d$ boxes, each ball with probability $p_i$ into box $i \in \{1,\ldots, d\}$, independently from all the other balls. What is the probability that there exist $\ell$ adjacent boxes in which together lie more than $k$ balls? 
Formally, if we turn to compute the probability of the complement: Let $N\sim \mathrm{M}_{n,p}$ be a multinomially distributed random variable. Task: Compute \[ \P\left(N_1+\ldots+ N_{\ell}\le k ,\ldots, N_{d-\ell+1}+\ldots+ N_{d} \le k \right)\] 
In this paper, we derive an algorithm that allows fast computation of this probability.

Such probabilities are needed as p-values for tests that check data on clusters. For example: Let $n=500$ patients arrive at a clinic in $d=365$ days. We compute the probability that there exist three successive days in which together more than $15$ patients arrive. From the line for $k=15$ in Table~\ref{tab:ScanEx} on page~\pageref{tab:ScanEx} below, we get the approximate value $1-0.9979961=0.0020039$ with an absolute error less than $10^{-7}$. As this probability is so small we would, if the described event occurs, reject the hypothesis that the patients arrived independently and hence suspect that there must be a reason for this cluster.

The support $D=\{x \in \N_0^d : x_1+\ldots+x_d=n \}$ of the multinomial distribution $\mathrm{M}_{n,p}$ is finite. Hence we could compute the desired probability as follows: For each $x \in D$ with $x_1+\ldots+x_\ell \le k, \ \ldots ,\ x_{d-\ell+1}+\ldots+x_d \le k$ compute the probability $\P(N=x)=
n!/(x_1!\ldots x_d!)p_1^{x_1}\ldots p_d^{x_d}$ and sum up these values. But because the support $D$ is large, this procedure takes  much time.
For example: If $n = 15, d=12$ it took $4$ 
seconds to compute the probability $\pi:=\P(N_1+N_2+N_3\le 5,\ldots,N_{d-2}+N_{d-1}+N_d \le 5)$ on a $3$ GHz desktop pc, for $n= 15, d= 25$ it already took $8$-$9$ hours.

To derive a faster method in Sections 2 and 3 of this paper, we use a fact already utilized by Corrado \cite{Corrado-Paper}, namely that the multinomially distributed random variable $N$ is a Markov increment, see Section 4. 
In this paper, a {\bf Markov increment} is a vector $(Y_1,\ldots,Y_d)$ of discretely distributed random variables with values in a group $({\cal X},\cdot)$ with the property that $(Y_1,Y_1\cdot Y_2,\ldots,Y_1\cdots Y_d)$ is a Markov chain.
Our method actually works for Markov increments in this generality. For example, the computation of the probability $\pi$ for $n= 15, d= 25$ with the new method takes less than one second.

In Sections 5 and 6 we turn to computer-implementations of our algorithm within the IEEE-754-standard \cite{IEEE754Standard} for floating point computer arithmetic. The floating point
number systems according to the IEEE-754-standard that usual computers work with have the following properties: The exact result of an operation on two floating point numbers, e.g. addition, need not be a floating point number again. In that case, the computer returns a floating point number that is as close as possible to the exact result. The difference between the returned value and the exact result is called {\bf rounding error}. Because of rounding errors, computed values, e.g. probabilities, are usually just approximations for the exact values and the goodness of the approximation is not known. One can switch the rounding mode of the machine in such a way, that in every operation it returns the minimal floating point number which is greater or equal than the exact result. This ``rounding up" mode can be used to compute upper bounds for the exact value, if only positive numbers occur and only additions and multiplications are performed. In the same way, per ``rounding down" mode, lower bounds can be computed. Thus one gets an interval whose bounds are floating point numbers and in which the exact value is known to lie. The accuracy of the approximations can easily be estimated, because the two bounds of the interval are known. In Section \ref{practsect}, we present an implementation of our algorithm within {\tt R}. For definiteness, we assume that all computations are done in double-precision according to the IEEE-754-standard. We analyze the accuracy of the {\tt R}-implementation and compare it to the best possible accuracy in IEEE-Double-Precision-computations of probabilities, which we examine in Section 5. 

To sum up: This work extends Corrado`s by clarifying the underlying Markov increment structure, by allowing the computation of scan probabilities and by providing rigorous numerical bounds.

\section{An algorithm that computes rectangle probabilities for Markov increments}\label{AlgoKapitel}
We derive an algorithm that computes rectangle probabilities for Markov increments. It is based on the following recursion formula:

\begin{theo} Let $Y=(Y_k)_{k=1}^d$ be Markov increment of a Markov chain $(X_k)_{k=1}^d$ which takes values in a group $({\cal X},\cdot)$. 
Let $A_1,\ldots, A_d \subset {\cal X}$ be countable sets. Then the probabilities 
\[ p(k,x):= \P(X_k = x, Y_1 \in A_1, \ldots , Y_k \in A_k) \] for $k \in \{1,\ldots, d\}$ and $x \in {\cal X}$ fulfill the recursion
\begin{equation} p(k,x)= \sum_{y \in A_k} \P(X_k = x \ | \ X_{k-1}=xy^{-1})p(k-1,xy^{-1}) \label{Recursion}\end{equation} for $k\ge 2$. Here and throughout, we use the convention $\P(A|B)=\P(A\cap B)/\P(B) := 0$ if $\P(B)=0$.
\end{theo}
\begin{proof}
The functions $f_k: {\cal X}^2 \to {\cal X}$ defined by $f_k(x_1,x_2)=x_1^{-1}x_2$ have the property that $Y_k=f_k(X_{k-1},X_k)$ and $f_k(\cdot, x ) $ is bijective for every $x \in {\cal X}$. Using this (which is actually all we need, so the method works not only for Markov increments but actually for any functions of two successive states of a Markov chain having the above bijectivity property) and writing $g_k(x,\cdot):=f_k(\cdot,x)^{-1}$, we get: 
\[ \lefteqn{\P(X_k=x, Y_1 \in A_1,\ldots, Y_k \in A_k)}\\
&=& \sum_{y \in A_k} \P(X_k=x,Y_k=y,Y_1\in A_1,\ldots, Y_{k-1} \in A_{k-1})\\
&=& \sum_{y \in A_k} \P(X_k=x,X_{k-1}=g_k(x,y),Y_1 \in A_1, \ldots, Y_{k-1}\in A_{k-1})\\
&=& \sum_{y \in A_k} \P(X_k=x|X_{k-1}=g_k(x,y))\\
& & \hspace{5 ex}\times\, \P(X_{k-1}=g_k(x,y),Y_1 \in A_1, \ldots, Y_{k-1}\in A_{k-1})
\] In the last step the Markov property was used.
\end{proof}

From the recursion formula we can derive the following algorithm that computes the probability $\P(Y_1 \in A_1,\ldots,Y_d \in A_d)$. Let $A_1,\ldots,A_d$ be finite, so that we get a finite algorithm. 
\vspace{0.2cm}

\fbox{\parbox{\linewidth}{
{\bf Algorithm A:}
\begin{enumerate}
\item For every $x \in A_1$ compute the value $p(1,x)=\P(X_1 = x)$
\item For every $k \in \{2,\ldots, d\}$:\\
For every $x \in A_1\cdot \ldots \cdot A_k$ compute the value $p(k,x)$ with formula  \eqref{Recursion}
\item Compute\\ $\P(Y_1 \in A_1,\ldots,Y_d \in A_d)=\displaystyle{\sum_{x \in A_1\cdot \ldots \cdot A_d}} \P(X_d=x , Y_1 \in A_1,\ldots,Y_d \in A_d)$
\end{enumerate}}}

\vspace{0.2cm}
\noindent Here, let $A_1 \cdots A_n := \{ a_1\cdots a_n : a_1 \in A_1,\ldots, a_n \in A_n \}$, 
if ${\cal X}$ is a group and $A_1,\ldots, A_n \subset {\cal X}$.

\section{Computing rectangle scan probabilities for Markov increments}
In this section we describe how to compute a rectangle scan probability
\[q:=P(Y_1\cdot\ldots\cdot Y_\ell \in A_1, \ldots, Y_{d-\ell+1}\cdot\ldots \cdot Y_d \in A_{d-\ell+1})\]
 for a Markov increment $Y$.

We use the following obvious and well-known lemma:
\begin{lem}
Let ${\cal X}$ be a countable set and $(X_k)_{k=1}^d$ an ${\cal X}$-valued Markov chain. Let $W_k := (X_{k},\ldots, X_{k+\ell-1})$. Then $(W_k)_{k=1}^{d-\ell+1}$ is an ${\cal X}^\ell$-valued Markov chain with transition probabilities
\[ \P(W_{k+1}=w\ |\  W_k=v) = \P(X_{k+\ell } =w_{\ell}\ |\  X_{k+\ell-1 } =v_{\ell})
\label{uewkeitentupel}
\]
for $v,w \in {\cal X}^\ell$ with $\P(W_k=v)>0$ and $v_2=w_1,\ldots,v_\ell = w_{\ell -1}$.
\end{lem}

The desired rectangle scan probability for the Markov increment $Y$ can be written as a rectangle probability for the increment $V$ of $W$: If we set $B_k:=\{(y_1,\ldots,y_\ell) \in {\cal X}^\ell | y_1\cdot \ldots\cdot y_\ell \in A_k\}$
we have
\[q= \P(V_1 \in B_1,\ldots, V_{d-\ell+1}\in B_{d-\ell+1}) \]
because $V_k= (X_{k}-X_{k-1},\ldots,X_{k+\ell-1}-X_{k+\ell})$ for $k \in \{2,\ldots,d-\ell+1\}$.

The sets $B_1,\ldots,B_{d-\ell +1}$ are possibly infinite so the Algorithm A from the last section would not work. But if there exist finite sets $M_1,\ldots,M_{d-\ell+1}\subset {\cal X}^\ell$ with 
\[\P(V_1 \in B_1,\ldots, V_{d-\ell+1}\in B_{d-\ell+1})  =\P(V_1 \in M_1,\ldots, V_{d-\ell+1}\in M_{d-\ell+1}) \]   
we can apply the Algorithm A and thus are able to compute the desired probability.

Example: If ${\cal X}=(\Z,+)$ and $Y$ is a Markov increment with $Y_1,\ldots,Y_{d} \ge 0$, then for finite sets $A_1,\ldots, A_{d-\ell+1} \subset \Z$ the probability 
\[ \P(Y_1+\ldots+Y_\ell \in A_1, \ldots, Y_{d-\ell+1}+\ldots+Y_d \in A_d ) \]
equals 
\[ \P( (Y_1,\ldots,Y_\ell)\in M_1, \ldots, (Y_{d-\ell+1},\ldots,Y_{d}) \in M_{d-\ell+1} ) \]
with $M_k:=\{ (y_1,\ldots, y_\ell) \in \N_0^{\ell} | y_1+\ldots+y_\ell \in A_k\}$, which are finite. 

\section{Examples for Markov increments: Multinomially and multivariate hypergeometrically distributed random vectors}\label{TransitionProbs}
By $\mathrm{b}_{n,p}(k)={n \choose k}p^k(1-p)^{n-k}$ we denote the binomial density with parameters $n \in \N$ and $p \in [0,1]$. By $\mathrm{h}_{n,r,b}(k)={r \choose k}{b\choose n-k}/{r+b \choose n}$ we denote the hypergeometrical density with parameters $r,b \in \N_0$ and $n \in \{1,\ldots, r+b\}$. 

Multinomially distributed random vectors as well as multivariate hypergeometrically distributed random vectors are Markov increments, hence the results from the last two sections are applicable in these cases. More precisely, we have the following two propositions, as easy calculation with density formulas and cancelling yield.
\begin{bsp}
Let $(N_1,\ldots,N_d) \sim \mathrm{M}_{n,p}$ be a multinomially distributed random variable and $S_k := \sum_{i=1}^k N_i$. Then $(S_1,\ldots,S_d)$ is a Markov chain with 
\[\P ( S_{k+1}= x | S_k = y ) = \mathrm{b}_{n-y,p_{k+1}/\sum_{i=k+1}^d p_i}(x-y)\]
\end{bsp}
\begin{bsp}
Let $(N_1,\ldots,N_d) \sim \mathrm{H}_{n,(m_1,\ldots,m_d)}$ be a multivariate hypergeometrically distributed random variable, i.e. $\P(N_1=k_1,\ldots,N_d=k_d)=$ \linebreak ${m_1 \choose k_1}\ldots {m_d \choose k_d}/{m_1+\ldots+m_d \choose n}$ for $k_1 \in \{0,\ldots,m_1\}, \ldots, k_d \in \{0,\ldots, m_d\}$ with $k_1+\ldots+k_d=n$,  and $S_k := \sum_{i=1}^k N_i$. Then $(S_1,\ldots,S_d)$ is a Markov chain with 
\[\P ( S_{k+1}= x | S_k = y ) = \mathrm{h}_{n-y,m_k,\sum_{i=k+1}^d m_i}(x-y)\]
\end{bsp}

\section{Definitions and notations for accuracy analyses of algorithms}
In this section we define terms we need to precisely describe the behaviour and the accuracy of numerical algorithms. For $M \subset {]0,\infty[}$ let $-M:= \{ - x : x \in M\}$ and $\pm M := M \cup (-M)$.

The {\bf IEEE-Double-Precision-Number-System} is the set
\[ \mathrm{IEEE\text{-}Double} :=  \pm F \cup \pm G \cup \{0,-\infty,\infty\}\] 
with
$ F:=\left\{m\cdot2^{e}: m \in \{2^{52},\ldots, 2^{53}-1\}, e \in \{-1074,\ldots, 971\} \right\}
$
and
$G:=\left\{k\cdot{2^{-1074}}: k \in \{1,\ldots,2^{52}-1\}\right\}$, compare \cite{IEEE754Standard}. The values $k/2^{52}$ in the definition of $G$ and the values $(m-2^{52})/2^{52}$ in the definition of $F$ are called {\bf mantissas} of the considered IEEE-Double-Numbers.
We consider the calculation of probabilities on computation systems that use IEEE-Double-Precision-Numbers.
Hence, every computable probability lies in the set
$ \mathrm{IEEE\text{-}Double} \cap [0,1] = G \cup \left\{m\cdot2^{e}: m \in \{2^{52},\ldots, 2^{53}-1\}, e \in \{-1074,\ldots, -53\} \right\}\cup\{0,1\}$, the minimal computable probability which is greater than zero is
$ \min\{ x \in \mathrm{IEEE\text{-}Double} : x> 0 \} = 2^{-1074} \approx 5\cdot 10^{-324} $
and the maximal computable probability which is less than one is
$ \max\{ x \in \mathrm{IEEE\text{-}Double} : x<1 \} = 1-2^{-53} \approx 1- 10^{-16} $.

We fix an object not belonging to the set $\mathrm{IEEE\text{-}Double}$, call it $\mathrm{NaN}$ for "Not a Number", and define the four operations
\[ \overline{+},\underline{+},\overline{\cdot}, \underline{\cdot} :\mathrm{IEEE\text{-}Double}\to \mathrm{IEEE\text{-}Double}\cup \{ \mathrm{NaN}\} \] 
For $x,y \in  \mathrm{IEEE\text{-}Double}$ and $\circ \in \{ +,\cdot\} $: 
\[ x \overline{\circ} y := \min \{ z \in \mathrm{IEEE\text{-}Double} : z \ge x \circ y \}\\
x \underline{\circ}y := \max \{ z \in \mathrm{IEEE\text{-}Double} : z \le x \circ y \}
\]
except for the following cases:
If  $x = 0$ and $y \in \{-\infty, \infty\}$ or $y = 0$ and $x \in \{-\infty, \infty\}$ then $x \overline{\cdot} y := x \underline{\cdot} y := \mathrm{NaN}$. If $x=-\infty$ and $y = \infty$  or $y = - \infty$ and $x = \infty $ then $x \overline{+} y := x \underline{+} y := \mathrm{NaN}$.
Note that the associative law does not hold for these four operations. For example let $a=-1, b=1, c=2^{-53}$, then we have $a\underline{+}(b\underline{+}c)=0\not= 2^{-53}=(a\underline{+}b)\underline{+}c$.

For the calculation of error bounds for the Algorithm A derived in Section~\ref{AlgoKapitel}, we use the following simple fact:
\begin{lem} Let $\circ \in \{+,\cdot\}$, $x,y \in {]0,\infty[}$ and $b_1,b_2,c_1,c_2 \in \mathrm{IEEE\text{-}Double}$ with $b_1\le x\le c_1$ and $b_2 \le y \le c_2$. Then
\[ b_1 \underline{\circ} b_2 \le x \circ y \le c_1 \overline{\circ} c_2 \]
\label{BasicPrinciple}
\end{lem}
For a quantitative analysis of the accuracy of computed probabilities we need to consider absolute and relative errors. 
For $p, \tilde{p}\in [0,1]$ we define the {\bf absolute error}
\[ e_{\mathrm{abs}}(p,\tilde{p}) := |p-\tilde{p}| \]
 and the {\bf relative error}
\[ e_{\mathrm{rel}}(p,\tilde{p}) :=  \max\left\{\frac{e_{\mathrm{abs}}(p,\tilde{p})}{p},\frac{e_{\mathrm{abs}}(1-p,1-\tilde{p})}{1-p}\right\}= \frac{|p-\tilde{p}|}{\min(p,1-p)}\]
 in the approximation of $p$ by $\tilde{p}$, with $\frac 00:= 0$ and $\frac x0:=\infty$ for $x>0$. 
For $a,b \in {[0,1]}$ with $a\le b$ and $\tilde{p}\in [a,b]$ we further define
the {\bf absolute error}
\[ e_{\mathrm{abs}}([a,b],\tilde{p}):= \max_{p \in [a,b]} e_{\mathrm{abs}}(p,\tilde{p}) =\max\{ b-\tilde{p}, \tilde{p}-a\}\]
and the {\bf relative error}  
\[ e_{\mathrm{rel}} ([a,b],\tilde{p}):= \max_{p \in [a,b]} e_{\mathrm{rel}}(p,\tilde{p}) \]
in the approximation of a probability which is known to lie in $[a,b]$ by $\tilde{p}$.
We get simple formulas for $e_{\mathrm{rel}} ([a,b],\tilde{p})$  in the following two cases. 
If $a,b \in {[0,1/2]}$ or $a,b \in {[1/2,1]}$ we have
\[e_{\mathrm{rel}} ([a,b],\tilde{p})= \max \{ e_{\mathrm{rel}}(a,\tilde{p}), e_{\mathrm{rel}}(b,\tilde{p})\}\]
Hence, if $a,b \in {]0,1/2]}$ we have 
\[ e_{\mathrm{rel}} ([a,b],\tilde{p})=\max\{ \frac{\tilde{p}-a}{a},\frac{b-\tilde{p}}{b}\} \]
and if $a,b \in {[1/2,1[}$ we have 
\[ e_{\mathrm{rel}}([a,b],\tilde{p})=\max\{ \frac{\tilde{p}-a}{1-a},\frac{b-\tilde{p}}{1-b}\}\] 
For accuracy measurements in interval calculations we use the following mini-max errors: 
\begin{defi}
For $a,b \in {[0,1]}$ with $a\le b$ we define the {\bf absolute error} 
\[ e_{\mathrm{abs}}([a,b]):= \min_{\tilde{p} \in [a,b]} e_{\mathrm{abs}}([a,b],\tilde{p})=e_{\mathrm{abs}}([a,b],\frac{a+b}2)=\frac{b-a}{2} \]
and the {\bf relative error}
\[ e_{\mathrm{rel}}([a,b]) := \min_{\tilde{p} \in [a,b]} e_{\mathrm{rel}}([a,b],\tilde{p})\]
in the approximation of a probability by the interval $[a,b]$. 
\end{defi}
Easy calculations yield the following formulas:
\begin{theo}
If $a,b \in {[0,1/2]}$ we have
\[ \forall {\tilde{p} \in [a,b]}: e_{\mathrm{rel}}([a,b],\tilde{p})\le e_{\mathrm{rel}}([a,b],\frac {2 ab}{a+b})=\frac{b-a}{b+a} \]
Hence
\[ e_{\mathrm{rel}}([a,b])=\frac {b-a}{b+a} \]
If $a,b \in {[1/2,1]}$ we have
\[ \forall {\tilde{p} \in [a,b]}: e_{\mathrm{rel}}([a,b],\tilde{p})\le e_{\mathrm{rel}}([a,b],\frac {a+b-2ab}{2-a-b})=\frac{b-a}{2-a-b} \]
Hence
\[ e_{\mathrm{rel}}([a,b])=\frac{b-a}{2-a-b} \]
\end{theo}Note that the absolute error $e_{\mathrm{abs}}([a,b])$ and the relative error $e_{\mathrm{rel}}([a,b])$ need not be reached simultaneously by one of the approximators. It need not be reached at all, as the following example illustrates.
\begin{bsp}
In Table~\ref{kleinesfehlerbsptab} we listed the errors $e_{\mathrm{abs}}([a,b],\tilde{p})$ and $ e_{\mathrm{rel}}([a,b],\tilde{p})$ for $[a,b]=[0.02,0.03]$ and different approximators $\tilde{p}$ . We see that $e_{\mathrm{abs}}([a,b])=0.005$ and $e_{\mathrm{rel}}([a,b])=1/5$. If we take the upper bound $\tilde{p}=b$ as approximator for the unknown probability $p$
, neither $e_{\mathrm{abs}}([a,b],\tilde{p})=e_{\mathrm{abs}}([a,b])$ is reached, nor $e_{\mathrm{rel}}([a,b],\tilde{p})=e_{\mathrm{rel}}([a,b])$.

\begin{center}\begin{table}[h]
\begin{tabular}{ccc}
$\tilde{p}$&$ e_{\mathrm{abs}}([a,b],\tilde{p})$&$ e_{\mathrm{rel}}([a,b],\tilde{p})$\\\hline
$2ab/(a+b)=0.024 $&$ 0.006 $&$ {\bf 1/5}$ \\
$(a+b)/2=0.025$ &$ {\bf 0.005}$ &$ 1/4$\\
$a$ & $0.01$ & $1/3$ \\
$b$ & $0.01$ & $1/2$ 
\end{tabular}
\caption{}
\label{kleinesfehlerbsptab}
\end{table}
\end{center}
If, for example, the unknown probability is $p = (3/10)^3= 0.027$, then the errors are as listed in Table~\ref{kleinesfehlerbsptab2}.
\begin{center}\begin{table}[h]
\begin{tabular}{ccc}
$\tilde{p}$&$ e_{\mathrm{abs}}(p,\tilde{p})$&$ e_{\mathrm{rel}}(p,\tilde{p})$\\\hline
$0.024 $&$ 0.003 $&$  3/27$ \\
$0.025$ &$  0.002$ &$  2/27$\\
$a$ & $0.007$ & $7/27$ \\
$b$ & $0.003$ & $3/27$ 
\end{tabular}
\caption{}
\label{kleinesfehlerbsptab2}
\end{table}
\end{center}
\label{kleinesfehlerbsp}
\end{bsp}
We study the maximal accuracy reachable in double-precision probability calculations:
\begin{defi}
The {\bf maximal accuracy} in a double-precision calculation of a probability $p \in [0,1]$ is $ e_{\mathrm{rel}}(p):= e_{\mathrm{rel}}(I(p))$
where $I(p):= [\max\{ x \in \mathrm{IEEE}\text{-}\mathrm{Double} : x \le p \} ,\min\{ x \in \mathrm{IEEE}\text{-}\mathrm{Double} : x \ge p \}]$ is the minimal interval containing $p$, whose endpoints are IEEE-Double-Precision-Numbers.
\end{defi}
Easy calculation yields
\[ e_{\mathrm{rel}}(p)=\left\{\begin{array} {cl}
\infty &0< p <2^{-1074} \text{ or } 1-2^{-53}<p<1 \\  
\frac{1}{2m+1} &  m \in \{1,\ldots, 2^{52}-1\},\\&  p \in 2^{-1074}\cdot{]m ,m+1[}  \text{ or } p \in  1-2^{-53}\cdot{]m ,m+1[} \\
\frac{1}{2m+1} &  m \in \{2^{52},\ldots, 2^{53}-1\}, e \in \{ -1022, \ldots, -2\},\\&p \in 2^{e-52}\cdot{]m ,m+1[}\\
0 & p \in \mathrm{IEEE-Double} \cap [0,1]
\end{array}\right.
\]

From the last formula it follows that 
\begin{itemize}\item in IEEE-Double floatingpoint arithmetic we are able to approximate probabilities in $[2^{-1074},1-2^{-53}]\cup\{0,1\}$ with finite relative error $e_{\mathrm{rel}}(p)$.
\item for $p \in [2^{-1022},1/2]$, the maximal accuracy satisfies $e_{\mathrm{rel}}(p)\le 1/(2^{53}+1)\approx 1.11\cdot 10^{-16}$.
\end{itemize}

\section{{\tt R} Implementation of Markov increment scan algorithms in interval arithmetic}\label{practsect}

{\texttt R} is an open source software for statistical computations.
We extended {\tt R} by a C-function that, as per C-Standard \cite{C-Standard} and IEEE-754-Standard \cite{IEEE754Standard}, allows the operations on IEEE-Double-Numbers which we defined in the previous section. 
We wrote an {\tt R}-program that implements the Algorithm A from Section \ref{AlgoKapitel} and uses the principle stated in Lemma \ref{BasicPrinciple} to compute bounds for rectangle scan probabilities for Markov increments. We implemented the multinomial and multivariate hypergeometric transition probabilities, as described in Section \ref{TransitionProbs}. In a last step the resulting {\tt R}-implementation of Algorithm A sets the returned value to $1$, if the original return value is greater than $1$.

\subsection{Examples}\label{ExampleSubsect}
For $N \sim \mathrm{M}_{n,p}$ with $n=500$, $d=365$, $p=(1/d,\ldots, 1/d)$ and $k \in \{4,\ldots,32\}$ we computed an upper bound $\overline{p}$ and a lower bound $\underline{p}$ for the probability $\P( \max_{i=1}^{d-2} (N_i+N_{i+1}+N_{i+2}) \le k)$ by an {\tt R}-implementation of the Algorithm A from Section \ref{AlgoKapitel}. 
In Table~\ref{tab:ScanEx} we tabulate the computed bounds $\underline{p},\overline{p}$ and analyze their accuracy. 
Numbers written in typewriter font are hexadecimal.  
The coloumn titled ``$\mathrm{approx}$'' gives the known decimal digits of a value of the ``probability representation number system'' $T$, that lies nearest to the exact value. The probability representation number system $T$ consists of all numbers with 7 decimal digits without leading zeros or nines. We use the notation $.0^x$ as an abbreviation for a decimal point followed by $x$ zeros, analogously $.9^x$. The symbol ? appearing in a number means that the following digits are not exactly known. 

The value $e_{\mathrm{abs}}$ resp. $e_{\mathrm{rel}}$ is the minimal upper bound for $e_{\mathrm{abs}}([\underline{p},\overline{p}])$ resp. $e_{\mathrm{rel}}([\underline{p},\overline{p}])$ 
which has the form $c\cdot 10^k$ where $c$ has $3$ significant digits and $k \in \Z$.

Thus, the line with $k=15$ means that  the probability $\P( \max_{i=1}^{d-2} (N_i+N_{i+1}+N_{i+2}) \le 15)$ lies in the interval $[\underline{p},\overline{p}]$ with 
\[ \overline{p}&=&\texttt{1.fef956911fe58}\cdot 2^{-1}\\&=&(1+15\cdot 16^{-1}+14\cdot 16^{-2}+\ldots + 8 \cdot 16^{-13})\cdot 2^{-1}\\ &=&0.99799604913273309847454584087245166301727294921875  \\\underline{p}&=&\texttt{1.fef95690c7eda}\cdot 2^{-1}\\&=& (1+15\cdot 16^{-1}+\ldots + 10\cdot 16^{-13})\cdot 2^{-1}\\&=&0.9979960490927297644958571254392154514789581298828125 \] 
with all equalities exact.
The minimal upper bound for $e_{\mathrm{abs}}([\underline{p},\overline{p}])$ 
which has the form $c\cdot 10^k$ where $c$ has $3$ significant digits and $k \in \Z$
is $2.01\cdot 10^{-11}$ and the minimal upper bound for $e_{\mathrm{abs}}([\underline{p},\overline{p}])$ which has this form is $ 9.99\cdot 10^{-9}$. 
A value of the number system $T$ which is nearest to the exact probability is $ 0.9979961$. As the numbers of the system $T$ in the interval $[0.001,0.9989999]$ differ by $10^{-7}$, just knowing the approximate value we can infer that the absolute error in this approximation is less than $10^{-7}$.

\begin{table}
{\scriptsize
\begin{tabular}{c||c|c|c||l}
$k$ &\multicolumn{1}{|c|}{$\underline{p},\overline{p}$} &  $e_{\mathrm{abs}}$ & $e_{\mathrm{rel}}$&$\mathrm{approx}$\\\hline\hline
${\bf 4}$ & \multicolumn{1}{|c|}{$0$}& \multicolumn{1}{|c|}{$0$}& $ 0$ & ${\bf 0}$ \\ \hline
${\bf5}$  &$\begin{array}{c}\texttt{1.1c5df1e1a1f83}\cdot 2^{-178}\\ \texttt{1.1c5df1e171043}\cdot 2^{-178}   \end{array}  $ & $ 5.82\cdot 10^{-65}$ & $ 2.01\cdot 10^{-11}$&${\bf.0^{53}28993}$\\\hline
${\bf6}$ &$\begin{array}{c} \texttt{1.b826f22f10057}\cdot 2^{-67}\\ \texttt{1.b826f22ec43c3}\cdot 2^{-67} \end{array}  $ &$ 2.34\cdot 10^{-31}$&$ 2.01\cdot 10^{-11}$&${\bf.0^{19}11651}$ \\\hline
${\bf7}$  &$\begin{array}{c} \texttt{1.b71c492587c97}\cdot 2^{-27} \\ \texttt{1.b71c49253c2df}\cdot 2^{-27}  \end{array}  $ & $ 2.57\cdot 10^{-19}$ & $ 2.01\cdot 10^{-11}$&${\bf.0^{7}12780}$\\\hline
${\bf8}$  &$\begin{array}{c}\texttt{1.98b8351d76fbd}\cdot 2^{-11} \\ \texttt{1.98b8351d309cf}\cdot 2^{-11}  \end{array}  $ & $ 1.57\cdot 10^{-14}$ & $ 2.01\cdot 10^{-11}$&${\bf.0^377957}$\\\hline
${\bf9}$  &$\begin{array}{c} \texttt{1.0f0230ce6f8a1}\cdot 2^{-4}  \\ \texttt{1.0f0230ce40e15}\cdot 2^{-4} \end{array}  $ & $ 1.33\cdot 10^{-12}$ & $ 2.01\cdot 10^{-11}$&${\bf .0661642}$\\\hline
${\bf10}$ &$\begin{array}{c} \texttt{1.826e2adb7befd}\cdot 2^{-2} \\ \texttt{1.826e2adb39686}\cdot 2^{-2} \end{array}  $ & $ 7.57\cdot 10^{-12}$ & $ 2.01\cdot 10^{-11}$&${\bf.3773734}$ \\\hline
${\bf11}$  &$\begin{array}{c} \texttt{1.7131cf887a229}\cdot 2^{-1} \\ \texttt{1.7131cf883a935}\cdot 2^{-1} \end{array}  $ & $ 1.45\cdot 10^{-11}$ & $ 5.19\cdot 10^{-11}$&${\bf.7210832}$\\\hline
${\bf12}$  &$\begin{array}{c} \texttt{1.ce576094ddb84}\cdot 2^{-1} \\ \texttt{1.ce5760948e1f6}\cdot 2^{-1} \end{array}  $ & $ 1.81\cdot 10^{-11}$ & $ 1.87\cdot 10^{-10}$&${\bf.9030104}$\\\hline
${\bf13}$ &$\begin{array}{c} \texttt{1.f1162301d80ec}\cdot 2^{-1} \\ \texttt{1.f1162301827ae}\cdot 2^{-1} \end{array}  $ & $ 1.95\cdot 10^{-11}$ & $ 6.69\cdot 10^{-10}$ &${\bf.9708720}$\\\hline
${\bf14}$ &$\begin{array}{c} \texttt{1.fbef9498b0df9}\cdot 2^{-1} \\ \texttt{1.fbef9498596d7}\cdot 2^{-1} \end{array}  $ & $ 1.99\cdot 10^{-11}$ & $ 2.51\cdot 10^{-9}$&${\bf.9920622}$ \\\hline
${\bf15}$  &$\begin{array}{c} \texttt{1.fef956911fe58}\cdot 2^{-1} \\ \texttt{1.fef95690c7eda}\cdot 2^{-1} \end{array}  $ & $ 2.01\cdot 10^{-11}$ & $ 9.99\cdot 10^{-9}$&${\bf.9979961}$\\\hline
${\bf16}$  &$\begin{array}{c} \texttt{1.ffc1fbbfd6e58}\cdot 2^{-1} \\ \texttt{1.ffc1fbbf7ecb1}\cdot 2^{-1} \end{array}  $ & $ 2.01\cdot 10^{-11}$ & $ 4.24\cdot 10^{-8}$&${\bf.9^352685}$\\\hline
${\bf17}$ &$\begin{array}{c} \texttt{1.fff23b0d23a3c}\cdot 2^{-1} \\ \texttt{1.fff23b0ccb810}\cdot 2^{-1} \end{array}  $ & $ 2.01\cdot 10^{-11}$ & $ 1.91\cdot 10^{-7}$&${\bf.9^389495}$ \\\hline
${\bf18}$  &$\begin{array}{c} \texttt{1.fffd1d22cb527}\cdot 2^{-1} \\ \texttt{1.fffd1d22732da}\cdot 2^{-1} \end{array}  $ & $ 2.01\cdot 10^{-11}$ & $ 9.11\cdot 10^{-7}$&${\bf.9^477980}$\\\hline
${\bf19}$ &$\begin{array}{c} \texttt{1.ffff6d5024936}\cdot 2^{-1} \\ \texttt{1.ffff6d4fcc6e4}\cdot 2^{-1} \end{array}  $ & $ 2.01\cdot 10^{-11}$ & $ 4.59\cdot 10^{-6}$&${\bf.9^556284}$ \\\hline
${\bf20}$  &$\begin{array}{c} \texttt{1.ffffe4570f39a}\cdot 2^{-1} \\ \texttt{1.ffffe456b7146}\cdot 2^{-1} \end{array}  $ & $ 2.01\cdot 10^{-11}$ & $ 2.44\cdot 10^{-5}$&${\bf .9^617567}$\\\hline
${\bf21}$ &$\begin{array}{c} \texttt{1.fffffb08bd13c}\cdot 2^{-1} \\ \texttt{1.fffffb0864ee9}\cdot 2^{-1} \end{array}  $ & $ 2.01\cdot 10^{-11}$ & $ 1.36\cdot 10^{-4}$&${\bf.9^68520?}$ \\\hline
${\bf22}$ &$\begin{array}{c} \texttt{1.ffffff264f47d}\cdot 2^{-1} \\ \texttt{1.ffffff25f7228}\cdot 2^{-1} \end{array}  $ & $ 2.01\cdot 10^{-11}$ & $ 7.91\cdot 10^{-4}$ &${\bf.9^774?}$\\\hline
${\bf23}$  &$\begin{array}{c} \texttt{1.ffffffdc79315}\cdot 2^{-1} \\ \texttt{1.ffffffdc210c0}\cdot 2^{-1} \end{array}  $ & $ 2.01\cdot 10^{-11}$ & $ 4.83\cdot 10^{-3}$&${\bf.9^86?}$\\\hline
${\bf24}$  &$\begin{array}{c} \texttt{1.fffffffa913ba}\cdot 2^{-1} \\ \texttt{1.fffffffa39167}\cdot 2^{-1} \end{array}  $ & $ 2.01\cdot 10^{-11}$ & $ 3.08\cdot 10^{-2}$&${\bf.9^9?}$\\\hline
${\bf25}$  &$\begin{array}{c} \texttt{1.ffffffff53a50}\cdot 2^{-1} \\ \texttt{1.fffffffefb7fe}\cdot 2^{-1} \end{array}  $ & $ 2.01\cdot 10^{-11}$ & $2.04\cdot 10^{-1}$&${\bf.9^9?}$\\\hline
${\bf26}$  &$\begin{array}{c} 1 \\ \texttt{1.ffffffffb44b7}\cdot 2^{-1} \end{array}  $ & $1-\underline{p}$ & $\infty$&${\bf.9^{10}?}$\\\hline
${\bf27}$  &$\begin{array}{c} 1 \\ \texttt{1.ffffffffcf373}\cdot 2^{-1} \end{array}  $ & $1-\underline{p}$ & $\infty$&${\bf.9^{10}?}$\\\hline
${\bf28}$  &$\begin{array}{c} 1 \\ \texttt{1.ffffffffd2fd3}\cdot 2^{-1} \end{array}  $ & $1-\underline{p}$ & $\infty$&${\bf.9^{10}?}$\\\hline
${\bf29}$ &$\begin{array}{c} 1 \\ \texttt{1.ffffffffd37fa}\cdot 2^{-1} \end{array}  $ & $1-\underline{p}$ & $\infty$ &${\bf.9^{10}?}$\\\hline
${\bf30}$  &$\begin{array}{c} 1 \\ \texttt{1.ffffffffd3908}\cdot 2^{-1} \end{array}  $ & $1-\underline{p}$ & $\infty$&${\bf.9^{10}?}$\\\hline
${\bf31}$  &$\begin{array}{c} 1 \\ \texttt{1.ffffffffd392a}\cdot 2^{-1} \end{array}  $ & $1-\underline{p}$ & $\infty$&${\bf.9^{10}?}$\\\hline
${\bf32}$  &$\begin{array}{c} 1 \\ \texttt{1.ffffffffd392a}\cdot 2^{-1} \end{array}  $ & $1-\underline{p}$ & $\infty$&${\bf.9^{10}?}$\\\hline
\end{tabular}
\caption{Upper and lower bounds $\overline{p},\underline{p}$ for $\P( \max_{i=1}^{d-2} N_i+N_{i+1}+N_{i+2} \le k)$ with $N \sim \mathrm{M}_{n,p}$, $n=500$, $d=365$, $p=(1/d,\ldots, 1/d)$ and $k \in \{4,\ldots,32\}$. For details, see Subsection \ref{ExampleSubsect}.}
\label{tab:ScanEx}
}
\end{table}

\subsection{Remarks on numerical computations of multinomial probabilities}
\subsubsection{Relative error of complement probabilities}
In the preceding section we computed the distribution function of a multinomial scan statistic.
For several applications, e.g. multinomial scan test, we want to compute the upper distribution function instead. If we compute its values from the complements $\P( \max_{i=1}^{d} N_i+N_{i+1}+N_{i+2} \ge k)=1-\P( \max_{i=1}^{d} N_i+N_{i+1}+N_{i+2} \le k-1)$  in exact arithmetic, for example by using a suitable software, there is no increse of error. If we do automatic computation of the complement in IEEE-Double-Precision-Number-System, then the error increases for small probabilities. Then, we are not able to approximate probabilities less then $10^{-16}$ with a finite relative error. Compare Table~\ref{tab:ScanExco}. Here, the relative error increases for small probabilities as well as for big probabilities. Small complements of probabilities are lost.
In general, one should try to avoid developing algorithms that complement the computed probability at the end. An algorithm that computes the complement is not equivalent to a direct one.

\begin{table}
{\scriptsize
\begin{tabular}{c||l|c|c||l}
$k$ &\multicolumn{1}{|c|}{$\underline{p},\overline{p}$} &  $e_{\mathrm{abs}}$ & $e_{\mathrm{rel}}$&$\mathrm{approx}$\\\hline\hline
${\bf5}$  &\multicolumn{1}{|c|}{$1$} & $ 0$ & $ 0$&$1$\\\hline
${\bf6}$ &$\begin{array}{c} \texttt{1} \\ \texttt{1.fffffffffffff}\cdot 2^{-1} \end{array}  $ &$ 1-\underline{p}$&$ \infty$&${\bf.9^{15}?}$ \\\hline
${\bf7}$  &$\begin{array}{c} \texttt{1} \\ \texttt{1.fffffffffffff}\cdot 2^{-1}\end{array}  $ & $ 1-\underline{p}$ & $ \infty$&${\bf.9^{15}?}$\\\hline
${\bf8}$  &$\begin{array}{c}  \texttt{1.ffffff9238edc}\cdot 2^{-1} \\ \texttt{1.ffffff9238edb}\cdot 2^{-1} \end{array}  $ & $ 5.55\cdot 10^{-17}$ & $ 4.34\cdot 10^{-9}$&${\bf.9^787220}$\\\hline
${\bf9}$  &$\begin{array}{c}  \texttt{1.ff99d1f2b8b3e}\cdot 2^{-1}\\ \texttt{1.ff99d1f2b8a24}\cdot 2^{-1}\end{array}  $ & $ 1.57\cdot 10^{-14}$ & $ 2.01\cdot 10^{-11}$&${\bf .9^322042}$\\\hline
${\bf10}$ &$\begin{array}{c}  \texttt{1.de1fb9e637e3e}\cdot 2^{-1}\\ \texttt{1.de1fb9e6320eb}\cdot 2^{-1} \end{array}  $ & $ 1.33\cdot 10^{-12}$ & $ 2.01\cdot 10^{-11}$&${\bf.9338358}$ \\\hline
${\bf11}$  &$\begin{array}{c} \texttt{1.3ec8ea92634bd}\cdot 2^{-1}\\ \texttt{1.3ec8ea9242081}\cdot 2^{-1} \end{array}  $ & $ 7.57\cdot 10^{-12}$ & $ 2.01\cdot 10^{-11}$&${\bf.6226266}$\\\hline
${\bf12}$  &$\begin{array}{c} \texttt{1.1d9c60ef8ad96}\cdot 2^{-2}\\ \texttt{1.1d9c60ef0bbae}\cdot 2^{-2} \end{array}  $ & $ 1.45\cdot 10^{-11}$ & $ 5.19\cdot 10^{-11}$&${\bf.2789168}$\\\hline
${\bf13}$ &$\begin{array}{c}  \texttt{1.8d44fb5b8f050}\cdot 2^{-4}\\ \texttt{1.8d44fb59123e0}\cdot 2^{-4} \end{array}  $ & $ 1.81\cdot 10^{-11}$ & $ 1.87\cdot 10^{-10}$ &${\bf.0969896}$\\\hline
${\bf14}$ &$\begin{array}{c}  \texttt{1.dd3b9fcfb0a40}\cdot 2^{-6}\\ \texttt{1.dd3b9fc4fe280}\cdot 2^{-6} \end{array}  $ & $ 1.95\cdot 10^{-11}$ & $ 6.69\cdot 10^{-10}$&${\bf.0291280}$ \\\hline
${\bf15}$  &$\begin{array}{c} \texttt{1.041ad9e9a4a40}\cdot 2^{-7}\\ \texttt{1.041ad9d3c81c0}\cdot 2^{-7} \end{array}  $ & $ 1.99\cdot 10^{-11}$ & $ 2.51\cdot 10^{-9}$&${\bf.0079377}$\\\hline
${\bf16}$  &$\begin{array}{c} \texttt{1.06a96f3812600}\cdot 2^{-9}\\ \texttt{1.06a96ee01a800}\cdot 2^{-9} \end{array}  $ & $ 2.01\cdot 10^{-11}$ & $ 9.99\cdot 10^{-9}$&${\bf.0020040}$\\\hline
${\bf17}$ &$\begin{array}{c}  \texttt{1.f0220409a7800}\cdot 2^{-12}\\ \texttt{1.f0220148d4000}\cdot 2^{-12}\end{array}  $ & $ 2.01\cdot 10^{-11}$ & $ 4.24\cdot 10^{-8}$&${\bf.0^347315}$ \\\hline
${\bf18}$  &$\begin{array}{c} \texttt{1.b89e668fe0000}\cdot 2^{-14}\\ \texttt{1.b89e5b8b88000}\cdot 2^{-14}\end{array}  $ & $ 2.01\cdot 10^{-11}$ & $ 1.91\cdot 10^{-7}$&${\bf.0^310505}$\\\hline
${\bf19}$ &$\begin{array}{c}  \texttt{1.716ec66930000}\cdot 2^{-16}\\ \texttt{1.716e9a56c8000}\cdot 2^{-16} \end{array}  $ & $ 2.01\cdot 10^{-11}$ & $ 9.11\cdot 10^{-7}$&${\bf.0^422020}$ \\\hline
${\bf20}$  &$\begin{array}{c} \texttt{1.2560672380000}\cdot 2^{-18}\\ \texttt{1.255fb6d940000}\cdot 2^{-18} \end{array}  $ & $ 2.01\cdot 10^{-11}$ & $ 4.59\cdot 10^{-6}$&${\bf .0^54371?}$\\\hline
${\bf21}$ &$\begin{array}{c} \texttt{1.ba948eba00000}\cdot 2^{-21}\\ \texttt{1.ba8f0c6600000}\cdot 2^{-21} \end{array}  $ & $ 2.01\cdot 10^{-11}$ & $ 2.44\cdot 10^{-5}$&${\bf.0^6824?}$ \\\hline
${\bf22}$ &$\begin{array}{c} \texttt{1.3de6c45c00000}\cdot 2^{-23}\\ \texttt{1.3dd0bb1000000}\cdot 2^{-23} \end{array}  $ & $ 2.01\cdot 10^{-11}$ & $ 1.36\cdot 10^{-4}$ &${\bf.0^614?}$\\\hline
${\bf23}$  &$\begin{array}{c} \texttt{1.b411bb0000000}\cdot 2^{-26}\\ \texttt{1.b361706000000}\cdot 2^{-26} \end{array}  $ & $ 2.01\cdot 10^{-11}$ & $ 7.91\cdot 10^{-4}$&${\bf.0^7253?}$\\\hline
${\bf24}$  &$\begin{array}{c} \texttt{1.1ef7a00000000}\cdot 2^{-28}\\ \texttt{1.1c36758000000}\cdot 2^{-28} \end{array}  $ & $ 2.01\cdot 10^{-11}$ & $ 4.83\cdot 10^{-2}$&${\bf.0^841?}$\\\hline
${\bf25}$  &$\begin{array}{c} \texttt{1.71ba640000000}\cdot 2^{-31}\\ \texttt{1.5bb1180000000}\cdot 2^{-31} \end{array}  $ & $ 2.01\cdot 10^{-11}$ & $ 3.08\cdot 10^{-2}$&${\bf.0^96?}$\\\hline
${\bf26}$  &$\begin{array}{c} \texttt{1.0480200000000}\cdot 2^{-33}\\ \texttt{1.58b6000000000}\cdot 2^{-34} \end{array}  $ & $ 2.01\cdot 10^{-11}$ & $2.04\cdot 10^{-1}$&${\bf.0^{9}?}$\\\hline
\end{tabular}}\caption{Upper and lower bounds $\overline{p},\underline{p}$ for $\P( \max_{i=1}^{d-2} N_i+N_{i+1}+N_{i+2} \ge k)$ with $N \sim \mathrm{M}_{n,p}$, $n=500$, $d=365$, $p=(1/d,\ldots, 1/d)$ and $k \in \{5,\ldots,26\}$.}
\label{tab:ScanExco}
\end{table}

The maximal accuracy with respect to complementation of probabilities is defined as
\[ e_{\mathrm{rel,c}}(p):= \max(e_{\mathrm{rel}}(p),e_{\mathrm{rel}}(1-p))\]
Easy calculation yields
\[ e_{\mathrm{rel}}(p)=\left\{\begin{array} {cl}
\infty &0< p <2^{-53} \text{ or } 1-2^{-53}<p<1 \\  
\frac{1}{2m+1} & m \in \{1,\ldots, 2^{52}-1\}, \\ & p \in 2^{-53}\cdot{]m ,m+1[} \text{ or } p \in  1-2^{-53}\cdot{]m ,m+1[} \\
0 & p \in \mathrm{IEEE-Double} \cap [0,1]
\end{array}\right.
\]

\subsubsection{Computation Time and Space}\label{TimeSpace}
Besides the accuracy of the algorithm, there are two other problems that matter: Time and space needed to compute the probability. 

The implementation of the multinomial scan algorithm we made needs to store $2* {n+\ell \choose \ell}$ Double-Precision-Numbers. Each Double-Precision-Number needs $8$ Bytes.
For example, for the scan width $\ell = 3$, on a computer with 16 GByte memory, we were able to compute Scan-Probabilities for up to approximately $n=1700$ in double-precision. Using Single-Precision-Numbers, which take only $4$ Bytes, we could compute up to $n=2150$, but the accuracy is worse than in double-precision computations, as Table~\ref{tab:ScanExSingle} demonstrates. 
The {\bf IEEE-Single-Precision-Number-System} is the set
\[ \mathrm{IEEE\text{-}Single} :=  \pm F \cup \pm G \cup \{0,-\infty,\infty
\}\]
with
$ F:=\left\{m*2^{e}: m \in \{2^{23},\ldots, 2^{24}-1\}, e \in \{-149,\ldots, 104\} \right\}$
and 
$G:=\left\{k*{2^{-149}}: k \in \{1,\ldots,2^{23}-1\}\right\}$,
compare \cite{IEEE754Standard}.
In the third coloumn of Table~\ref{tab:ScanExSingle} we listed the first digits of the computed bounds $\underline{p},\overline{p}$ in decimal format.

\begin{table}
{\scriptsize
\begin{tabular}{c||c|l|c|c||l}
$k$ &\multicolumn{1}{|c|}{$\underline{p},\overline{p}$}&\multicolumn{1}{|c|}{$\underline{p},\overline{p}$} &  $e_{\mathrm{abs}}$ & $e_{\mathrm{rel}}$&$\mathrm{approx}$\\\hline\hline
${\bf 4}$ & \multicolumn{1}{|c|}{$0$}& \multicolumn{1}{|c|}{$0$}& \multicolumn{1}{|c|}{$0$}& $ 0$ & $ 0$ \\ \hline
${\bf5}$  &$\begin{array}{c}\texttt{1.974c00}\cdot 2^{-135} \\0\   \end{array}  $ &$\begin{array}{l}.0^{40}3652... \\ {\centering 0}\   \end{array}  $  & $ 1.83\cdot 10^{-41}$ & $ 1.04\cdot 10^{-2}$&$.0^{40}?$\\\hline
${\bf6}$ &$\begin{array}{c} \texttt{1.bcc5a4}\cdot 2^{-67}\\ \texttt{1.b39300}\cdot 2^{-67} \end{array}  $ &$\begin{array}{l}.0^{19}1177... \\.0^{19}1152...\   \end{array}  $  &$ 1.22\cdot 10^{-22}$&$ 1.04\cdot 10^{-2}$&$.0^{19}11?$ \\\hline
${\bf7}$  &$\begin{array}{c} \texttt{1.bbb862}\cdot 2^{-27} \\ \texttt{1.b28b40}\cdot 2^{-27}  \end{array}  $ &$\begin{array}{l}.0^712913... \\.0^712646... \   \end{array}  $  & $1.34\cdot 10^{-10}$ & $ 1.04\cdot 10^{-2}$&$.0^712?$\\\hline
${\bf8}$  &$\begin{array}{c}\texttt{1.9d02a2}\cdot 2^{-11} \\ \texttt{1.947834}\cdot 2^{-11}  \end{array}  $ &$\begin{array}{l}.0^378775... \\.0^377146... \   \end{array}  $  & $ 8.15\cdot 10^{-6}$ & $ 1.04\cdot 10^{-2}$&$.0^37?$\\\hline
${\bf9}$  &$\begin{array}{c} \texttt{1.11da84}\cdot 2^{-4}  \\ \texttt{1.0c30d0}\cdot 2^{-4} \end{array}  $ &$\begin{array}{l}.0668587... \\.0654762...\   \end{array}  $  & $ 6.91\cdot 10^{-4}$ & $ 1.04\cdot 10^{-2}$&$.06?$\\\hline
${\bf10}$ &$\begin{array}{c} \texttt{1.867cac}\cdot 2^{-2} \\ \texttt{1.7e699a}\cdot 2^{-2} \end{array}  $  &$\begin{array}{l}.3813349... \\.3734497...\   \end{array}  $ & $ 3.94\cdot 10^{-3}$ & $ 1.04\cdot 10^{-2}$&$.3?$\\\hline
${\bf11}$  &$\begin{array}{c} \texttt{1.7511fc}\cdot 2^{-1} \\ \texttt{1.6d5b2a}\cdot 2^{-1} \end{array}  $  &$\begin{array}{l}.7286528... \\.7135861...\   \end{array}  $ & $ 7.53\cdot 10^{-3}$ & $ 2.7\cdot 10^{-2}$&$.7?$\\\hline
${\bf12}$  &$\begin{array}{c} \texttt{1.d331e6}\cdot 2^{-1} \\ \texttt{1.c988cc}\cdot 2^{-1} \end{array}  $  &$\begin{array}{l}.9124900... \\.8936218...\   \end{array}  $ & $ 9.43\cdot 10^{-3}$ & $ 9.7\cdot 10^{-2}$&$.?$\\\hline
${\bf13}$ &$\begin{array}{c} \texttt{1.f64e04}\cdot 2^{-1} \\ \texttt{1.ebeb16}\cdot 2^{-1} \end{array}  $  &$\begin{array}{l}.9810639... \\.9607779...\   \end{array}  $ & $ 1.01\cdot 10^{-2}$ & $ 3.49\cdot 10^{-1}$ &$.9?$\\\hline
${\bf14}$ &$\begin{array}{c} 1 \\ \texttt{1.f6a7a6}\cdot 2^{-1} \end{array}  $  &$\begin{array}{c}1 \\.9817478...\   \end{array}  $ & $1-\underline{p}$ & $\infty$&$.9?$ \\\hline
${\bf15}$  &$\begin{array}{c} 1 \\ \texttt{1.f9a956}\cdot 2^{-1} \end{array}  $  &$\begin{array}{c}1 \\.9876200...\   \end{array}  $ & $1-\underline{p}$ & $\infty$&$.9?$\\\hline
${\bf16}$  &$\begin{array}{c} 1 \\ \texttt{1.fa6fe6}\cdot 2^{-1} \end{array}  $  &$\begin{array}{c}1 \\.9891349...\   \end{array}  $ &$1-\underline{p}$ & $ \infty$&$.9?$\\\hline
${\bf17}$ &$\begin{array}{c}1 \\ \texttt{1.fa9fa0}\cdot 2^{-1} \end{array}  $  &$\begin{array}{c}1 \\.9894990...\   \end{array}  $ & $1-\underline{p}$ & $ \infty$&$.9?$ \\\hline
${\bf18}$  &$\begin{array}{c}1 \\ \texttt{1.faaa68}\cdot 2^{-1} \end{array}  $ &$\begin{array}{c}1 \\.9895813...\   \end{array}  $  &$1-\underline{p}$ & $ \infty$&$.9?$\\\hline
${\bf19}$ &$\begin{array}{c} 1 \\ \texttt{1.faacb6}\cdot 2^{-1} \end{array}  $ &$\begin{array}{c}1 \\.9895989...\   \end{array}  $  & $1-\underline{p}$& $ \infty$&$.9?$ \\\hline
${\bf20}$  &$\begin{array}{c} 1 \\ \texttt{1.faad2c}\cdot 2^{-1} \end{array}  $ &$\begin{array}{c}1 \\.9896024...\   \end{array}  $  & $1-\underline{p}$ & $ \infty$&$.9?$\\\hline
${\bf21}$ &$\begin{array}{c}1 \\ \texttt{1.faad3c}\cdot 2^{-1} \end{array}  $ &$\begin{array}{c}1 \\.9896029...\   \end{array}  $  &$1-\underline{p}$ & $ \infty$&$.9?$ \\\hline
${\bf22}$ &$\begin{array}{c} 1 \\ \texttt{1.faad40}\cdot 2^{-1} \end{array}  $ &$\begin{array}{c}1\\.9896030...\   \end{array}  $  & $1-\underline{p}$& $ \infty$ &$.9?$\\\hline
${\bf23}$  &$\begin{array}{c} 1 \\ \texttt{1.faad44}\cdot 2^{-1} \end{array}  $ &$\begin{array}{c}1\\.9896031...\   \end{array}  $  & $1-\underline{p}$& $ \infty$&$.9?$\\\hline
${\bf24}$  &$\begin{array}{c} 1 \\ \texttt{1.faad46}\cdot 2^{-1} \end{array}  $ &$\begin{array}{c}1 \\.9896032...\   \end{array}  $  & $1-\underline{p}$ & $ \infty$&$.9?$\\\hline
${\bf25}$  &$\begin{array}{c} 1 \\ \texttt{1.faad46}\cdot 2^{-1} \end{array}  $ &$\begin{array}{c}1\\.9896032...\   \end{array}  $  &$1-\underline{p}$ & $\infty$&$.9?$\\\hline
\end{tabular}}\caption{Upper and lower bounds $\overline{p},\underline{p}$ for $\P( \max_{i=1}^{d-2} N_i+N_{i+1}+N_{i+2} \le k)$ with $N \sim \mathrm{M}_{n,p}$, $n=500$, $d=365$, $p=(1/d,\ldots, 1/d)$ and $k \in \{4,\ldots,25\}$, computed in single-precision.}\label{tab:ScanExSingle}
\end{table}

The time that it takes to compute a rectangle scan probability for a multinomially distributed random vector in single precision does not differ much from the time it takes in double-precision, examples are listed in Table~\ref{tab:examplestimes}.

\begin{table}
\begin{tabular}{c|c|c|c|r|c|r}
$n$ & $d$ & $k$ & $\underline{p}$ (double) & time (double) & $\underline{p}$ (single) & time (single)$\phantom{\int_{X_x}}$\\\hline
$100$ & $365$ & $6$ & $0.9934578$ & 1 s& $0.9914927$& 1 s \\
$500$ &  $365$ & $13$ & $0.9708720$ & 1 min 27 s& $0.9607779$& 1 min 25 s\\
$1000$ & $365$ & $20$ & $0.9604324$ & 49 min 39 s& $0.9405573$& 10 min 19 s \\
$1500$ & $365$ & $27$ & $0.9739303$ &3 h 22 min 26 s& $0.9438554$& 2 h 44 min 00 s \\
$1700$ & $365$ & $29$ & $0.9610315$ &5 h 47 min 01 s& $0.9274842$&4 h 58 min 21 s\\
$1750$ & $365$ & $31$ & x & no computation possible & $0.9516879$& 6 h 36 min 48 s \\
$2150$ & $365$ & $37$ & x & no computation possible & $0.9507257$& 12 h 47 min 22 s\\
\end{tabular}\caption{Computation time for a lower bound $\underline{p}$ for  the scan probability  $\P( \max_{i=1}^{d-2} N_i+N_{i+1}+N_{i+2} \le k)$ for a random vector $(N_1,\ldots, N_d)$ with the multinomial distribution $\mathrm{M}_{n,p}$ with $p=(1/d,\ldots,1/d)$,  in single precision and in double precision. Details are described in Subsubsection~\ref{TimeSpace}}
\label{tab:examplestimes}
\end{table}

\subsection{Binomial Probabilities}
We use the following algorithm to compute the multinomial transition probabilities, that are binomial. 
\begin{verbatim}
double bnp(unsigned int k,unsigned int n, double p, double q){
if (2*k>n) return(bnp(n-k,n,q,p));
double f=1.0;
unsigned int j0=0,j1=0,j2=0;
while ( (j0<k) | (j1<k)| (j2<n-k) )
{ 
if( (j0<k) && (f<1) ) {j0++; f*= (double)(n-k+j0)/(double)j0;}
 else { if(j1<k) {j1++; f*= p;} else {j2++; f*= q;} }
}
return f;
}
\end{verbatim}
For upper bounds $p$ and $q$ and ``rounding up'' mode the algorithm calculates an upper bound for the exact binomial probability.
For lower bounds $p$ and $q$ and ``rounding down'' mode the algorithm calculates a lower bound for the exact binomial probability.

\subsection{Hypergeometric Probabilities}
We use the following algorithm to compute the multivariate hypergeometric transition probabilities, that are univariate hypergeometric.
Table~\ref{HypEx} contains the distribution function of the random variable $\max_{i=1}^{d-2} (N_i+N_{i+1}+N_{i+2})$ with $N\sim \mathrm{H}_{n,m}$ with $n=500$, $d=365$ and $m=(10,\ldots,10)$.
\begin{verbatim}
double hyp(int n, int r, int b, int k){
double f=1.0;
int j0=0,j1=0,j2=0;
while ( (j0<k)| (j1<n-k) | (j2<n) ){
if(f<1 && ( (j0<k) | (j1<n-k)) ){
if (j0<k) { f*=(double)(r-j0)/(j0+1);j0++;} 
else {if (j1<n-k) { f*=(double)(b-j1)/(j1+1);j1++;} 
else if (j2<n) {f*=(double)(r+b-j2)/(j2+1);j2++;}}
}
else if (j2<n) { f*=(double)(j2+1)/(r+b-j2);j2++;}
}
return f;
}
\end{verbatim}
In the ``rounding up'' mode the algorithm calculates an upper bound for the exact hypergeometric probability.
In the ``rounding down'' mode the algorithm calculates a lower bound for the exact hypergeometric probability.

\begin{table}
{\scriptsize
\begin{tabular}{c||c|c|c||l}
$k$ &\multicolumn{1}{|c|}{$[\underline{p},\overline{p}]$} &  $e_{\mathrm{abs}}$ & $e_{\mathrm{rel}}$&$\mathrm{approx}$\\\hline\hline
${\bf 4}$ & \multicolumn{1}{|c|}{$0$}& \multicolumn{1}{|c|}{$0$}& $ 0$ & ${\bf 0}$ \\ \hline
${\bf5}$  &$\begin{array}{c} \texttt{1.94a78cce6bf78}\cdot 2^{-160} \\ \texttt{1.94a78cce088a0}\cdot 2^{-160} \end{array}  $ & $ 3.09\cdot 10^{-59}$ & $ 2.86\cdot 10^{-11}$&${\bf.0^{47}10815}$\\\hline
${\bf6}$ &$\begin{array}{c} \texttt{1.0acc3dae78827}\cdot 2^{-55} \\ \texttt{1.0acc3dae36d0e}\cdot 2^{-55} \end{array}  $ &$ 8.29\cdot 10^{-28}$&$ 2.87\cdot 10^{-11}$&${\bf.0^{16}28926}$ \\\hline
${\bf7}$  &$\begin{array}{c} \texttt{1.591d6928456d6}\cdot 2^{-20} \\ \texttt{1.591d6927f05d0}\cdot 2^{-20} \end{array}  $ & $ 3.69\cdot 10^{-17}$ & $ 2.87\cdot 10^{-11}$&${\bf.0^{5}12856}$\\\hline
${\bf8}$  &$\begin{array}{c} \texttt{1.40ac4ad3593a9}\cdot 2^{-7} \\ \texttt{1.40ac4ad30a26f}\cdot 2^{-7} \end{array}  $ & $ 2.81\cdot 10^{-13}$ & $ 2.87\cdot 10^{-11}$&${\bf.0097862}$\\\hline
${\bf9}$  &$\begin{array}{c} \texttt{1.df885f4b6ceae}\cdot 2^{-3} \\ \texttt{1.df885f4af6a55}\cdot 2^{-3} \end{array}  $ & $ 1.91\cdot 10^{-11}$ & $ 2.87\cdot 10^{-11}$&${\bf.2341468}$\\\hline
${\bf10}$ &$\begin{array}{c} \texttt{1.546bd869a7f5e}\cdot 2^{-1} \\ \texttt{1.546bd86953fe9}\cdot 2^{-1} \end{array}  $ & $ 2.60\cdot 10^{-11}$ & $ 5.70\cdot 10^{-11}$&${\bf.6648853}$ \\\hline
${\bf11}$  &$\begin{array}{c} \texttt{1.cec1ebd5b5793}\cdot 2^{-1} \\ \texttt{1.cec1ebd543545}\cdot 2^{-1} \end{array}  $ & $ 2.81\cdot 10^{-11}$ & $ 2.70\cdot 10^{-10}$&${\bf.9038233}$\\\hline
${\bf12}$  &$\begin{array}{c} \texttt{1.f4e8088a29393}\cdot 2^{-1} \\ \texttt{1.f4e80889adab5}\cdot 2^{-1} \end{array}  $ & $ 2.86\cdot 10^{-11}$ & $ 1.30\cdot 10^{-9}$&${\bf.9783328}$\\\hline
${\bf13}$ &$\begin{array}{c} \texttt{1.fde26f4234a4c}\cdot 2^{-1} \\ \texttt{1.fde26f41b6dfc}\cdot 2^{-1} \end{array}  $ & $ 2.87\cdot 10^{-11}$ & $ 6.92\cdot 10^{-9}$ &${\bf.9958682}$\\\hline
${\bf14}$ &$\begin{array}{c} \texttt{1.ffa6780ca228e}\cdot 2^{-1} \\ \texttt{1.ffa6780c23f48}\cdot 2^{-1} \end{array}  $ & $ 2.87\cdot 10^{-11}$ & $ 4.20\cdot 10^{-8}$&${\bf.9^331693}$ \\\hline
${\bf15}$  &$\begin{array}{c} \texttt{1.fff314a41d498}\cdot 2^{-1} \\ \texttt{1.fff314a39f023}\cdot 2^{-1} \end{array}  $ & $ 2.87\cdot 10^{-11}$ & $ 2.91\cdot 10^{-7}$&${\bf.9^401433}$\\\hline
${\bf16}$  &$\begin{array}{c} \texttt{1.fffe5ec7c001c}\cdot 2^{-1} \\ \texttt{1.fffe5ec741b7c}\cdot 2^{-1} \end{array}  $ & $ 2.87\cdot 10^{-11}$ & $ 2.31\cdot 10^{-6}$&${\bf.9^487566}$\\\hline
${\bf17}$ &$\begin{array}{c} \texttt{1.ffffd2049693f}\cdot 2^{-1} \\ \texttt{1.ffffd20418497}\cdot 2^{-1} \end{array}  $ & $ 2.87\cdot 10^{-11}$ & $ 2.10\cdot 10^{-5}$&${\bf.9^58629?}$ \\\hline
${\bf18}$  &$\begin{array}{c} \texttt{1.fffffb9535338}\cdot 2^{-1} \\ \texttt{1.fffffb94b6e8e}\cdot 2^{-1} \end{array}  $ & $ 2.87\cdot 10^{-11}$ & $ 2.18\cdot 10^{-4}$&${\bf.9^6868?}$\\\hline
${\bf19}$ &$\begin{array}{c} \texttt{1.ffffffa1dc0a3}\cdot 2^{-1} \\ \texttt{1.ffffffa15dbfb}\cdot 2^{-1} \end{array}  $ & $ 2.87\cdot 10^{-11}$ & $ 2.61\cdot 10^{-3}$&${\bf.9^78?}$ \\\hline
${\bf20}$  &$\begin{array}{c} \texttt{1.fffffff9717b1}\cdot 2^{-1} \\ \texttt{1.fffffff8f330a}\cdot 2^{-1} \end{array}  $ & $ 2.87\cdot 10^{-11}$ & $ 3.63\cdot 10^{-2}$&${\bf .9^9?}$\\\hline
${\bf21}$ &$\begin{array}{c} \texttt{1.ffffffffd3bf1}\cdot 2^{-1} \\ \texttt{1.ffffffff55749}\cdot 2^{-1} \end{array}  $ & $ 2.87\cdot 10^{-11}$ & $ 5.88\cdot 10^{-1}$&${\bf.9^{10}?}$ \\\hline
${\bf22}$ &$\begin{array}{c} 1 \\ \texttt{1.ffffffffbb782}\cdot 2^{-1} \end{array}  $ & $1-\underline{p}$ & $\infty$ &${\bf.9^{10}?}$\\\hline
${\bf23}$  &$\begin{array}{c} 1 \\ \texttt{1.ffffffffc0de3}\cdot 2^{-1} \end{array}  $ & $1-\underline{p}$& $\infty$&${\bf.9^{10}?}$\\\hline
${\bf24}$  &$\begin{array}{c}1 \\ \texttt{1.ffffffffc11b4}\cdot 2^{-1} \end{array}  $ & $1-\underline{p}$ &$\infty$&${\bf.9^{10}?}$\\\hline
${\bf25}$  &$\begin{array}{c} 1 \\ \texttt{1.ffffffffc11d9}\cdot 2^{-1} \end{array}  $ & $1-\underline{p}$ & $\infty$&${\bf.9^{10}?}$\\\hline
${\bf26}$  &$\begin{array}{c} 1 \\ \texttt{1.ffffffffc11d9}\cdot 2^{-1} \end{array}  $ & $1-\underline{p}$ & $\infty$&${\bf.9^{10}?}$\\\hline
\end{tabular}}\caption{Upper and lower bounds $\overline{p},\underline{p}$ for $\P( \max_{i=1}^{d-2} N_i+N_{i+1}+N_{i+2} \le k)$ with $N \sim \mathrm{H}_{n,m}$, $n=500$, $d=365$, $m=(10,\ldots, 10)$ and $k \in \{4,\ldots,26\}$.}\label{HypEx}
\end{table}

\section*{Acknowledgements}
The author wishes to thank Lutz Mattner for many valuable hints and discussions. The author wishes to thank Christoph Tasto and Todor Dinev for proofreading.

\end{document}